\documentclass[9pt]{entcs} 
\usepackage{entcsmacro}
\usepackage[utf8x]{inputenc} 
\usepackage{graphicx}
\usepackage[all]{xy}
\usepackage{amsmath}
\usepackage{amssymb}
\usepackage{color}
\usepackage{cite}

\usepackage{graphics}
\usepackage{float}
\usepackage{pstricks, pst-node, pst-tree}
\usepackage{stmaryrd}
\usepackage{url}
\usepackage{soul}
\usepackage[title]{appendix}

\newcommand{\nat}{{\mathbb N}}

\newcommand{\pre}{\mathsf{pre}}
\newcommand{\amodel}{\ensuremath{\mathsf{M}}}

\newcommand{\Actions}{\ensuremath{\mathsf{S}}}
\newcommand{\action}{\ensuremath{\mathsf{s}}}
\newcommand{\actiona}{\action}
\newcommand{\actionb}{\ensuremath{\mathsf{t}}}

\newcommand{\imp}{\rightarrow}
\newcommand{\eq}{\leftrightarrow}
\newcommand{\et}{\wedge}

\newcommand{\Et}{\bigwedge}

\newcommand{\knows}{K}

\newcommand{\agent}{a}

\def\lastname{Benevides \& Lima}
\begin{document}
\begin{frontmatter}
  \title{Dynamic Epistemic Logic with Communication Actions} 
  \author{Mario Roberto Folhadela Benevides\thanksref{myemail}}
  \address{PESC/COPPE - Inst.\ de Matem\'atica/DCC\\ Federal University of Rio de Janeiro\\ Rio de Janeiro, Rio de Janeiro}
  \author{Isaque Macalam Saab Lima\thanksref{coemail}}
  \address{PESC/COPPE\\ Federal University of Rio de Janeiro\\ Rio de Janeiro, Rio de Janeiro}
  \thanks[myemail]{Email:\href{mailto:mario@cos.ufrj.br} 
      {\texttt{\normalshape mario@cos.ufrj.br}}} 
  \thanks[coemail]{Email:\href{mailto:isaque@cos.ufrj.br} 
      {\texttt{\normalshape isaque@cos.ufrj.br}}}
\begin{abstract} 
  This work proposes a Dynamic Epistemic Logic with Communication Actions that can be performed concurrently. Unlike Concurrent Epistemic Action Logic introduced by Ditmarsch, Hoek and Kooi \cite{hwbcdel}, where the concurrency mechanism is the so called {\it true concurrency}, here we use an approach based on process calculus, like CCS and CSP, and Action Models Logic. Our approach makes possible the proof of soundness, completeness and decidability, different from the others approaches. We present an axiomatization and show that the proof of soundness, completeness and decidability can be done using a reduction method.
\end{abstract}
\end{frontmatter}
\section{Introduction}\label{intro}

  Multi-Agent Epistemic Logic has been investigated in Computer Science \cite{halpern95} to represent and reason about agents (or groups of agents') knowledge and beliefs. Dynamic Logic  aims to reason about actions (programs) and their effects \cite{harel}. Dynamic Epistemic Logic \cite{hwb} is conceived to reason about actions that change agents (or groups of agents') epistemic state, i.e., actions which change agent's knowledge and beliefs.
  
  The first Dynamic Epistemic Logic was proposed independently by \cite{plazza} and \cite{jelle} it is called Public Announcement Logic(PAL) . There are many other approaches but the one that is used in this work is the Action Model Logic proposed by  \cite{baltag04,baltag98}. 
  
  Concurrent Dynamic Epistemic Logic was introduced in \cite{hwbcdel} and it was intended to extend Epistemic Action Logic proposed by Van Ditmarsch in \cite{hcdel} with concurrent epistemic actions. In this extension they use a mechanism to deal with concurrency called "true concurrency" which is inspired on the Concurrent  Propositional Dynamic Logic proposed by Peleg in \cite{DPEL87}. An interesting work, entitled Logics of Communication and Knowledge, presented in \cite{floor}, proposes a framework for modeling message passing situations that combines properties of dynamic epistemic semantics and history-based approaches, which consists of Kripke models with records of sent messages in their valuations. Another work that inspired us to represent communication actions as private epistemic action is \cite{jelle}.

  Example: Consider that there are two students waiting for a message from a teacher to send back the homework and that one student does not know if the other received or responded the message. To represent this we need to model the following actions: teacher sending the message (send action), each student receiving (receive action) and responding (response action) the message independently. We also need to guarantee that: the receive action can not be performed before the send action, the response action can not be performed before the receive action and the students actions can be performed concurrently. Can we model this using Action Models Logic? Since this is a very small example one can argue that this can done by using pre conditions and non deterministic choice to model all the possible paths. Now imagine the same situation with 100 students. It would be  not so easy to model.

   This work proposes a way to deal with concurrency and communication with Dynamic Epistemic Logic.  We use an approach based on action models and process calculus, like CCS and CSP, which allow us to prove soundness, completeness and decidability. Different from  \cite{hwbcdel}, that implements concurrency on top of Epistemic Action Logic, we extends Action Models to deal with concurrency and communication. The proofs of soundness, completeness and decidability can be done using a reduction method.
   
    
    In order to facilitate the proof of soundness, completeness, and decidability we restricted our concurrency approach. We do not deal with "true concurrency" like in \cite{hwbcdel}. Instead, we adopt the interleaving (non-deterministic choices of all possible paths) approach used in process algebras like CCS and CSP. Since we are based on Action Models we can use the pre-conditions to restrict actions that must be executed after another action. We do not deal with Common Knowledge, because this would make the proofs a little more tricky.
    
  
  In sections \ref{mael}, \ref{aml} and \ref{eal} we give a brief introduction to Multi-agent Epistemic Logic, Action Model Logic and Concurrent Dynamic Epistemic Logic. Next we present the Dynamic Epistemic Logic that we propose in this paper. The last section is the conclusion.

\section{Multi-Agent Epistemic Logic}\label{mael}

  This section presents the Multi-Agent Epistemic Logic ${\bf S5_a}$. All the definitions and theorems of this section are based on \cite{hwb}.
  
  \subsection{Language and Semantics}

    \begin{definition}\label{el-lang} 
      The Epistemic language consists of a countable set $\Phi$ of proposition symbols,  a finite set ${\cal A}$ of agents,  a modality $K_a$ for each agent $a$ and the boolean connectives $\neg$ and $\land$. The formulas are defined as follows:
      \[\varphi ::= p \mid \top \mid \neg \varphi \mid \varphi_1 \wedge \varphi_2  \mid K_a \varphi  \]

      where $p \in \Phi$,  $a \in {\cal A}$.
    \end{definition}

    \begin{definition}\label{el-frame} 
      A multi-agent epistemic \emph{frame}  is a tuple $\mathcal{F}= (S, R_a)$ where:
      \begin{itemize}
        \item $S$ is a non-empty set of states;
        \item $R_a$ is a binary relation over $S$, for each agent $a \in {\cal A}$;
      \end{itemize}
    \end{definition}

    \begin{definition}\label{el-model} 
      A multi-agent epistemic \emph{model} is a pair $\mathcal{M}= ({\cal F}, {\bf V})$, where ${\cal F}$ is a  frame and ${\bf V}$  is a valuation 
      function ${\bf V} : \Phi \to 2^S$. We call a rooted  multi-agent epistemic model $(\mathcal{M}, s)$ an epistemic state.
    \end{definition}

    \begin{definition}\label{el-satisfaction} Given a multi-agent epistemic model  ${\cal M} = \langle (S, R_a), {\bf V} \rangle$. The notion of satisfaction ${\cal M},s \models \varphi$ is defined as follows:   
      \begin{description}
        \item[1.] $ {\cal M},s \models p \text{ iff } s \in {\bf V}(p)$
        \item[2.] $ {\cal M},s \models \neg \phi \text{ iff } {\cal M},s \not \models \phi $
        \item[3.] $ {\cal M},s \models \phi \land \psi \text{ iff } {\cal M},s \models \phi \text{ and } {\cal M},s \models \psi $
        \item[4.] $ {\cal M},s \models K_a \phi  \text{ iff } \text{for all } s' \in S: s R_a s' \Rightarrow {\cal M},s' \models \phi $
      \end{description}   
      
    \end{definition}

  \subsection{Axiomatization}\label{mael-ax}

    \begin{enumerate}
      \item  {\it All instantiations of propositional tautologies}, 
      \item   $K_a(\varphi \rightarrow \psi) \rightarrow (K_a\varphi \rightarrow K_a\psi)$,
      \item   $K_a \varphi \rightarrow \varphi$, 
      \item   $K_a \varphi \rightarrow K_a K_a \varphi~~~~~~~~~(+~introspection)$,
      \item   $\neg K_a \varphi \rightarrow K_a \neg K_a \varphi~~~~~(-~introspection)$,
    \end{enumerate}

    \centerline{\bf Inference Rules\/}
    M.P. $\varphi, \varphi\rightarrow \psi / \psi$ \qquad U.G. $\varphi/ K_a\varphi$ \qquad   

    \begin{theorem} ${\bf S5_a}$ is sound and complete w.r.t its semantics. \end{theorem}
    
    \begin{example}\label{el-card} This example is from \cite{hwb}.
    
      Suppose we have a card game with three cards: {\bf 0}, {\bf 1} and {\bf 2}, and three players {\bf a}, {\bf b} and {\bf c}. Each player receives a card and 
      do not know the other players cards. 
      
      We use proposition symbols $0_x, 1_x, 2_x$ for $x \in \{{\bf a}, {\bf b} , {\bf c}\}$ meaning ``player $x$ has 
      card {\bf 0}, {\bf 1} or {\bf 2}''. We name each state by the cards that each player has in that state, for instance $012$ is the state where player {\bf a} has
      card {\bf 0}, player {\bf b} has card {\bf 1} and player {\bf c} has card {\bf 2}\footnote{A state name underlined means current state}.
      The folowing epsitemic model repesents the epistemic state of each agent\footnote{We omitt the reflexive loops in the picture.}.     
      
      $Hexa1 = \langle (S, R), {\bf V} \rangle$:
      
      \begin{itemize}
        \item $S = \{ 012, 021, 102, 120, 201, 210 \}$
        \item $R = \{ (012,012), (012, 021), (021, 021), \dots \}$
        \item ${\bf V}(0_a) = \{012, 021\}$, ${\bf V}(1_a) = \{102, 120\}$, ...
      \end{itemize}     
    
      \begin{figure}[H]
        \centering
        \includegraphics[scale=0.25]{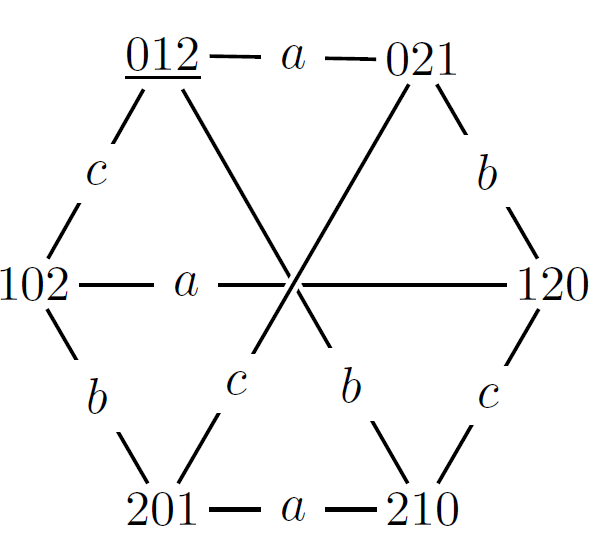}
        \caption{Epistemic Model $Hexa_1$}\label{hexa}
      \end{figure}
        
    \end{example}

\section{Action Models}\label{aml}

  All the definitions and theorems of this section are based on \cite{hwb}.
  
  \subsection{Language and Semantics}
    
    \begin{defn}\label{action-models}
      An action model $\amodel$ is a structure $\langle \Actions,\sim_a,\pre \rangle$, where:
      \begin{itemize}
      \item $\Actions$ is a \emph{finite} domain of action points or events;
      \item $\sim_{a}$ is an equivalence relation on $\Actions$, for each agent $a \in {\cal A}$;
      \item $\pre: \Actions \mapsto {\cal L}$ is a precondition function that assigns a
      precondition to each $\actiona\in\Actions$. 
      \end{itemize}
    
      Rooted action models is an action model with a distinguished  state $(\amodel,\actiona)$.
    
      Note that $\Actions$ is different from $S$, $\amodel$ is different from $\cal{M}$ and $\actiona$ is different from $s$.
    \end{defn}
    
    \begin{defn}\label{def-langactionmodel} 
      The Action Model language consists of a countable set $\Phi$ of proposition symbols,  a finite set ${\cal A}$ of agents, the boolean connectives $\neg$ and $\land$, a modality $K_a$ for each agent $a  \in {\cal A}$ and a modality $[\alpha]$ . The formulas are defined as follows:
      \[\varphi ::= p \mid \top \mid \neg \varphi \mid \varphi_1 \wedge \varphi_2  \mid K_a \varphi \mid [\alpha]\varphi,  \]
      \[ \alpha ::= (\amodel,\actiona) \mid \alpha_1;\alpha_2 \mid \alpha_1 \cup \alpha_2 \]
    
      where $p \in \Phi$,  $a \in {\cal A}$, $(\amodel,\actiona)$ a rooted action model and $\langle \alpha \rangle \eq \neg [\neg \alpha]$
    \end{defn}  
        
    \begin{defn}\label{product-action-models}
      Given an epistemic state $({\cal M},s)$ with ${\cal M} = \langle (S, R_a), {\bf V} \rangle$ and a rooted action model $(\amodel,\actiona)$ with $\amodel = \langle \Actions, \sim_a, \pre \rangle$. The result of executing $(\amodel,\actiona)$ in $({\cal M},s)$ is $({\cal M} \otimes \amodel,(s,\actiona))$ where ${\cal M} \otimes \amodel = \langle (S', R'_a), {\bf V}' \rangle$ such that:
      
      \begin{enumerate}
        \item $S' = \{(s,\actiona) ~such~that~ s \in S, \action \in \Actions, \text{ and } {\cal M}, s \models \pre(\actiona) \}$
        \item $(s,\actiona) R'_a (t,\actionb)  \text{ iff }  (s~R_a~t \text{ and } \actiona \sim_a \actionb)$
        \item $(s,\actiona) \in {\bf V}'(p) \text{ iff } s \in {\bf V}(p)$
      \end{enumerate}
    
    \end{defn}
    
    \begin{defn}\label{comp-action-models}{\bf Composition of rooted action models}
      
      Given rooted action models $(\amodel,\actiona)$ with $\amodel = \langle \Actions, \sim, \pre \rangle$ and $(\amodel',\actiona')$ with $\amodel' = \langle \Actions', \sim', \pre' \rangle$, their composition is the action model $(\amodel ; \amodel', (\actiona,\actiona'))$ with $\amodel ; \amodel' = \langle \Actions'', \sim'', \pre'' \rangle$:
    
      \begin{itemize}
        \item $\Actions'' = \{(\action,\actiona') ~such~that~ \action\in\Actions, \action' \in \Actions'$ \}
        \item $(\action,\action') \sim''_\agent (\actionb,\actionb')  \text{ iff }  (\action \sim_\agent \actionb \text{ and } \actiona' \sim'_\agent \actionb')$
        \item $\pre''(\action,\action') = \langle (\amodel,\action)\rangle \pre'(\action')$
      \end{itemize}
      
    \end{defn}

    \begin{defn}\label{satisfaction-am} Given a rooted epistemic state $({\cal M},s)$ with ${\cal M} = \langle (S, R_a), {\bf V} \rangle$ and a rooted action model $(\amodel,\actiona)$ with $\amodel = \langle     \Actions, \sim, \pre \rangle$. The notion of satisfaction ${\cal M},s \models \varphi$ extends from \ref{el-satisfaction} and is defined as follows
    
      \begin{description}
        \item[1,2,3, 4] as in definition \ref{el-satisfaction}
        \item[5.] ${\cal M},s \models [(\amodel,\actiona)]\phi \text{ iff }  {\cal M},s \models \pre(\actiona) \Rightarrow {\cal M}\otimes\amodel,(s,\actiona) \models \phi$
        \item[6.] $\llbracket \alpha \cup \beta \rrbracket   \text{ iff }  \llbracket \alpha \rrbracket \cup \llbracket \beta \rrbracket$
        \item[7.] $\llbracket(\amodel,\actiona);(\amodel',\actiona')\rrbracket  \text{ iff }  (\amodel ; \amodel', (\actiona,\actiona'))$ ~ Composition of action models
      \end{description}
      Where $\llbracket . \rrbracket$ is the interpretation on a action model.
    
    \end{defn}
    
  \subsection{Axiomatization}\label{ax-del}
    
    \begin{itemize}
        
      \bigskip     

      \centerline{\bf Epistemic Logic Axioms\/}
      \bigskip
    
       \item[]  {\it Axioms (i), (ii), (iii), (iv) and (v) of  section \ref{mael-ax}},
           
      \bigskip
    
      \centerline{\bf Action Model Logic Axioms\/}    
      \bigskip
       
      \item[(vi)] $[(\amodel,\actiona)]p \eq (\pre(\actiona) \imp p)$,    
      \item[(vii)] $[(\amodel,\actiona)] \neg \phi \eq (\pre(\actiona) \imp \neg [(\amodel,\actiona)] \phi)$  
      \item[(viii)] $[(\amodel,\actiona)] (\phi \et \psi) \eq ([(\amodel,\actiona)]\phi \et [(\amodel,\actiona)]\psi)$     
      \item[(ix)] $[(\amodel,\actiona)] \knows_\agent \phi \eq (\pre(\actiona) \imp \Et_{\actiona \sim_\agent \actionb} \knows_\agent [(\amodel,\actionb)] \phi)$
      \item[(x)] $[(\amodel,\actiona)] [(\amodel',\actiona')] \phi \eq [(\amodel,\actiona);(\amodel',\actiona')] \phi$    
      \item[(xi)] $[(\amodel,\actiona) \cup (\amodel',\actiona')]  \phi \eq [(\amodel,\actiona)] \phi \land  [(\amodel',\actiona')] \phi$ 
    \end{itemize}
    
    \centerline{\bf Inference Rules\/}    
      \smallskip

      M.P. $\varphi, \varphi\rightarrow \psi / \psi$ \qquad U.G.  $\varphi/ K_a \varphi$  \qquad $\varphi/ [\alpha] \varphi$ 
      \bigskip
    
      Every formula in the language of action model logic without common
      knowledge is equivalent to a formula in the language of epistemic logic \cite{hwb}.

    \begin{example} {\bf Continuation of example \ref{el-card}}\label{exampleActionModel}
    
      Suppose now agent {\bf a} wants to perform the action of showing her card to agent {\bf b}. In fact, we have three actions, agent {\bf a}  showing either card {\bf 0}, {\bf 1} or {\bf 2} to agent {\bf b}. Agents   {\bf a} and {\bf b} can distinguish between these three action but agent  {\bf c} cannot. This situation can be represented by the action model below.
    
    \begin{figure}[H]
    \centering
    \includegraphics[scale=0.25]{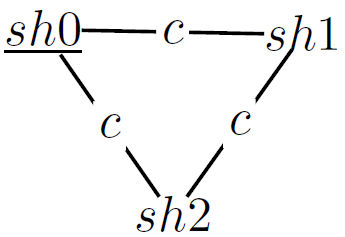} 
    \caption{Action Model for $show$}\label{hexaam}
    \end{figure}

    \begin{itemize}
      \item $\Actions=\{\mathsf{sh0,sh1,sh2}\}$
      \item $\sim_a \ =\{(\actiona,\actiona) \mid \actiona \in \Actions\}$
      \item $\sim_b \ =\{(\actiona,\actiona) \mid \actiona \in \Actions\}$
      \item $\sim_c \ =\Actions \times \Actions$
      \item $\pre(\mathsf{sh0})=0_a$
      \item $\pre(\mathsf{sh1})=1_a$
      \item $\pre(\mathsf{sh2})=2_a$
    \end{itemize}
    
      If agent {\bf a} performs the action of showing her card to agent {\bf b} on the epistemic model of example \ref{el-card}, we obtain:

    \begin{figure}[H]
      \centering
      \includegraphics[scale=0.25]{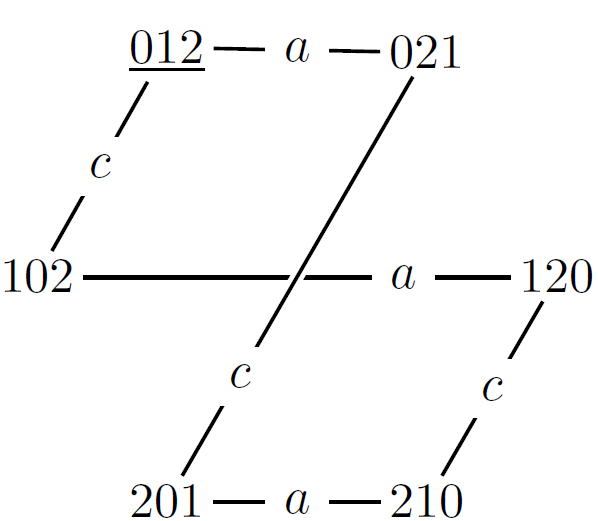}
      \caption{$Hexa_1$ After the Execution of $show$}\label{hexadepois1}
    \end{figure}
    
      This new epistemic model, shown in figure \ref{hexadepois1},  is obtained by the product of epistemic model of figure \ref{hexa} with the action model of figure \ref{hexaam}. It is important to notice that the number of states after the product is $18$ ($6\times3$), but most of them are thrown out because they do not satisfy the precondition.
     
    \end{example}

\section{Epistemic Actions and Concurrent Dynamic Epistemic Logic}\label{eal}

This section provides a brief introduction to the works presented in  \cite{hwbcdel} and \cite{hwb}.

Epistemic Actions is an extension of Multi-Agent Epismtemic Logic to deal with new information (updates), like Action Models, but it uses a different approach to deal with new information. Concurrent Dynamic Epistemic Logic proposes a way to deal with concurrency in Epistemic Actions.

    \subsection{Language and Semantics}
        \begin{defn}\label{def-langepistemicactions} 
            The Epistemic Actions language consists of a countable set $\Phi$ of proposition symbols,  a finite set ${\cal A}$ of agents, the boolean connectives $\neg$ and $\land$, a modality $K_a$ for each agent $a  \in {\cal A}$ and a modality $[\alpha]$ . The formulas and the actions are defined as follows:
            \[\varphi ::= p \mid \top \mid \neg \varphi \mid \varphi_1 \wedge \varphi_2  \mid K_a \varphi \mid [\alpha]\varphi,  \]
            \[ \alpha ::=  ?\alpha \mid L_{\cal B} \beta \mid (\alpha ! \alpha) \mid (\alpha \text{\textexclamdown} \alpha) \mid  (\alpha ; \beta) \mid (\alpha_1 \cup \alpha_2) \]

            where $p \in \Phi$,  $a \in {\cal A}$, ${\cal B} \subseteq {\cal A}$, $L$ stands for learning and $L_{\cal B} \beta$  means 'group ${\cal B}$ learn that $\beta$, $?\alpha$ is a test, $(\alpha ! \alpha)$ is called left local choice, $(\alpha \text{\textexclamdown} \alpha)$ is called right local choice, $(\alpha ; \beta)$ is sequential composition (first $\alpha$ then $\beta$), $(\alpha_1 \cup \alpha_2)$ is non-deterministic choice. 
        \end{defn}

      \begin{defn}\label{ea-semantic} Given the epistemic model  ${\cal M} = \langle S, \sim_a, V \rangle$ and the state $s \in S$. The notion of satisfaction ${\cal M},s \models \varphi$ extends from \ref{el-satisfaction} and is defined as follows
    
      \begin{description}
        \item[1,2,3, 4] as in definition \ref{el-satisfaction}
        \item[5.] ${\cal M},s \models [\alpha]\phi \text{~~~iff~~~} \text{for all } ({\cal M'},s'): ({\cal M},s) [\alpha]({\cal M'},s')  \text{ implies } ({\cal M'},s') \models \phi $
            \item[6.] $({\cal M},s) [?\phi]({\cal M'},s') \text{~~~iff~~~} {\cal M'} = \langle [\phi]_M , \emptyset , V \cap [\phi]_M \rangle \text{ and } s' = s $
            \item[7.] $({\cal M},s) [L_G \phi]({\cal M'},s') \text{~~~iff~~~} {\cal M'} = \langle S' , \sim' , V' \rangle \text{ and } ({\cal M},s)[\phi] s'$
            \item[8.] $\llbracket \alpha ; \alpha' \rrbracket \text{~~~ = ~~~} \llbracket \alpha \rrbracket \circ \llbracket \alpha' \rrbracket$
            \item[9.] $\llbracket \alpha \cup \alpha' \rrbracket \text{~ = ~~~} \llbracket \alpha \rrbracket \cup \llbracket \alpha' \rrbracket$
            \item[10.] $\llbracket \alpha ! \alpha' \rrbracket \text{~~~ = ~~~} \llbracket \alpha \rrbracket$
        \end{description}
      \end{defn}

            The Concurrent Dynamic Epistemic Logic language adds the concurrent execution operator to the actions of  Epistemic Actions language. The actions are defined as follows:
            
            \[ \alpha ::=  ?\alpha \mid L_{\cal B} \beta \mid (\alpha ! \alpha) \mid (\alpha \text{\textexclamdown}  \alpha) \mid  (\alpha ; \beta) \mid (\alpha_1 \cup \alpha_2) \mid (\alpha_1 \cap \alpha_2) \]

            where $(\alpha_1 \cap \alpha_2)$ represents a concurrent execution.

    \begin{example}\label{ea-cdel-example} In order to illustrate the use of the language of  Epistemic Actions, we consider the  card game presented in section \ref{aml}.
        
        The Epistemic Model is the same shown in figure \ref{hexa}. 
        
        The  action  of ``agent {\bf a} showing her card to agent {\bf b}" can be model as:
        \[ \begin{array}{lcl}
            (L_{(b)}?0_a \cup L_{(b)}?1_a \cup L_{(b)}?2_a);(L_{(a,b,c)}?(K_{b}0_a \vee K_{b}1_a \vee K_{b}2_a))
        \end{array} \]    
        This means that agent {\bf a} tells agent {\bf b} her card and after that all agents know that agent {\bf b} knows the card that agent {\bf a} holds.
      After performing this action, the resulting  epistemic model is the same as in  figure \ref{hexadepois1}.
        
    \end{example}

\section{Dynamic Epistemic Logic with Communication Actions}
 
  \subsection{Process Calculus} \label{pc}
    
    In this section, we propose a very small process (program) calculus for the  programs of Dynamic Epistemic Logic with Communication Actions (DELWCA). It is inspired by \cite{glabbleek}.
    
    Let ${\cal A} = \{1, ..., n \}$, denoted by $i, j ...$, be a finite set of agents,  AMS=$ \{ a_1,a_2,a_3\ldots\}$ be a finite set of action models and ${\cal N} = \{c_1,c_2,c_3,\ldots, \overline{c_1}, \overline{c_2}, \overline{c_3}, \ldots\}$ be a finite set of communication actions. As a convention, communication actions with one overline represent output and with no overlines represent an input. Communication actions can be combined to form a private action model, by joining an output communication action with its respective input  ( [$ c_1, \overline{c_1} $] = $a_1$ ). The action model resultant of the join of two communication actions is known as silent action, denoted by $\tau^{s}_{i,j}(.)$,  that can be interpreted as the result of a communication between agents $i$ and $j$\footnote{As silent actions $\tau^{s}_{i,j}(.)$ are interpreted as private action models, the index $s$ denotes the root of the action model $\tau^{s}_{i,j}(.)$.}.
    
    \begin{defn}\label{program-syntax} The language can be defined as follows.   
   
      $$\eta ::=\alpha \mid \alpha.\eta \mid \eta_1;\eta_2 \mid \eta_1 + \eta_2 , \text{ where } \alpha \in \text{AMS} \cup {\cal N}$$
      
      $$\pi ::= \eta \mid \beta. \pi \mid \pi_1;\pi_2 \mid \pi_1 + \pi_2  \mid \eta_1 \parallel \eta_2 \cdots \parallel \eta_n$$ where $n = |{\cal A}|$ and $\eta_i$ denotes the program performed by agent $i$.
      
      We use $\pi$ and $\eta$ to denote processes (programs) and $\alpha$ and $\beta$ to denote action models and communication actions.

      The \emph{prefix} operator $.$ denotes that the process will first perform the action $\alpha$ and then behave as $\pi$. The \emph{summation} (or \emph{nondeterministic choice}) operator $+$ denotes that the process will make a nondeterministic choice to behave as either $\pi_1$ or $\pi_2$. The \emph{parallel composition} operator $\parallel$ denotes that the processes $\eta_1, ..., \eta_n$, performed by agents $1,...,n$ respectively, may proceed independently or may communicate through a common channel. 
    \end{defn}

    We write $\pi \stackrel{\alpha}{\rightarrow} \pi'$ to express that the process $\pi$ can perform the action $\alpha$ and after that behave as $\pi'$. We write $\pi \stackrel{\alpha}{\rightarrow} \surd$ to express that the process $\pi$ successfully finishes after performing the action $\alpha$. A process  finishes when there is no possible action left for it to perform. For example, $\beta \stackrel{\beta}{\rightarrow} \surd$. When a process finishes inside a parallel composition, sequential composition or non-deterministic choice we write $\pi$ instead of $\pi| \surd$, $\pi ; \surd$ and $\pi + \surd$. We also write $\surd$ instead of  $\surd | \surd$.
    
    Like \cite{milner89} we need to restrict the agents to perform some actions. In our case we don't want to perform communication actions, but we can perform $\tau$ action which results from the combination of communication actions ($ \overline{a}, \overline{\overline{a}} $).
    
    The semantics of our process calculus can be given by the transition rules presented in table \ref{tab:semccs}, where $\pi$ and $\eta$  are process specifications, while $\pi'$ and $\eta'$ are process specifications or $\surd$. The $\tau^{s}_{i,j}(.)$ action  represents an  internal communication action from agent $i$ to agent $j$.
    
    \Large
    \begin{table}[H]
      \centering
      \begin{tabular}{|c|}
        \hline
        $\alpha \stackrel{\alpha}{\rightarrow} \surd$
        \\
        \hline
        $\alpha.\pi \stackrel{\alpha}{\rightarrow} \pi$ 
        \\
        \hline
        $\displaystyle{ \frac{\pi_1 \stackrel{\alpha}{\rightarrow} \pi_1'} {\pi_1 ; \pi_2 \stackrel{\alpha}{\rightarrow} \pi_1' ; \pi_2}}$
        \\
        \hline  
        $\displaystyle \frac{\pi_1 \stackrel{\alpha}{\rightarrow} \pi_1'} {\pi_1 + \pi_2 \stackrel{\alpha}{\rightarrow} \pi_1'}$
        \\
        \hline
        $\displaystyle \frac{\pi_2 \stackrel{\beta}{\rightarrow}\pi_2'} {\pi_1 + \pi_2 \stackrel{\beta}{\rightarrow}\pi_2'}$
        \\
        \hline
        $\displaystyle{ \frac{\eta_i \stackrel{\alpha}{\rightarrow} \eta_i'}{(\eta_1 \parallel ... \parallel \eta_i \parallel ... \parallel \eta_n) 
        \stackrel{\alpha}{\rightarrow} (\eta_1 \parallel ... \parallel \eta_i' \parallel ... \parallel \eta_n)}}, ~for~all~i,j \in {\cal A}$ 
        \\
        \hline
        $\displaystyle{ \frac{\eta_i \stackrel{\overline{c}}{\rightarrow} \eta_i', \eta_j \stackrel{c}{\rightarrow} \eta_j'}
        {(\eta_1 \parallel ... \parallel \eta_i \parallel ... \parallel \eta_j \parallel ... \parallel \eta_n) \stackrel{\tau_{i,j}(.)}{\rightarrow} 
        (\eta_1 \parallel ... \parallel \eta_i' \parallel ... \parallel \eta_j' \parallel ... \parallel \eta_n)}}, ~for~all~i,j \in {\cal A}$ 
        \\
        \hline
      \end{tabular}
      \caption{Transition Relation}
      \label{tab:semccs}
    \end{table}
    
    \normalsize
    
    \vspace{0.3cm}
       
    \begin{example}\label{el-card-3} Continuation of the card game example.
     
      Now suppose that the game is online and player {\bf a} sends a message $p$ to players {\bf b} and {\bf c}. So after the message $p$ players {\bf b} and {\bf c} know all the cards.
      This problem can be modeled as follows:
    
      \begin{itemize}
        \item $ \pi_1 = \overline{c_{ab}}(p);\overline{c_{ac}}(p) + \overline{c_{ac}}(p);\overline{c_{ab}}(p) $
        \item $ \pi_2 = c_{ab}(.).\beta $
        \item $ \pi_3 = c_{ac}(.).\gamma $
        \item $ \pi_{1\parallel2\parallel3} = (\pi_1 \parallel  \pi_2 \parallel \pi_3) $
      \end{itemize}
    
      Given the programs $\pi_1, \pi_2$ and $\pi_3$, initially we have two possible actions: communication between {\bf a} and {\bf b} or communication  between {\bf a} and {\bf c}. Suppose that the communication between {\bf a} and {\bf b} occurs first, then we will have two possible actions: communication between  {\bf a} and {\bf c} or action $\beta$ and so on ...

      We can represent this using parallel composition:
    
    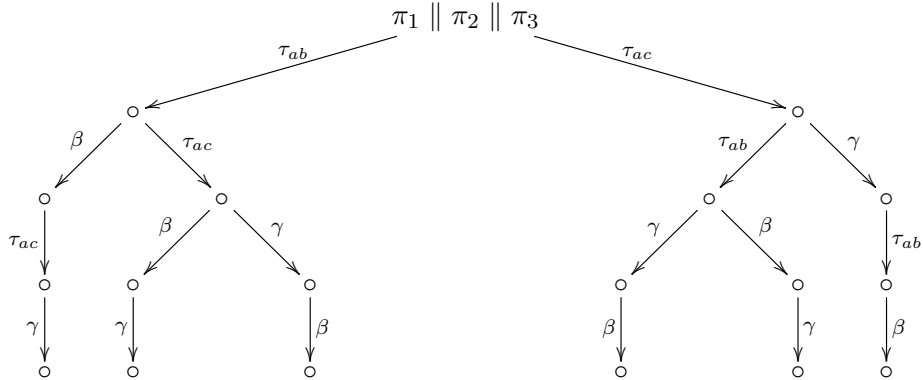
\begin{figure}[H]
      \centering
      
      $\xymatrix{
      &		&   &   &  	&\pi_1 \parallel \pi_2 \parallel \pi_3 \ar[llld]_{\tau_{ab}} \ar[rrrd]^{\tau_{ac}}	&			&			&   &    \\
      &   & \circ \ar[ld]_{\beta}	\ar[rd]^{\tau_{ac}}&    &		&	   	&   &   & \circ	\ar[ld]_{\tau_{ab}} \ar[rd]^{\gamma}  &     \\
      &  \circ \ar[d]_{\tau_{ac}} &  	&  \circ \ar[ld]_{\beta} \ar[rd]^{\gamma} &		&	   	&   & \circ \ar[ld]_{\gamma} \ar[rd]^{\beta}  &  &  \circ	\ar[d]^{\tau_{ab}}   \\
      &  \circ \ar[d]_{\gamma} & \circ \ar[d]_{\gamma}	&    &	\circ \ar[d]^{\beta}	&	   	&  \circ \ar[d]_{\beta} &   &  \circ \ar[d]^{\gamma} &  \circ	\ar[d]^{\beta}   \\
      &  \circ  & \circ	&    &	\circ	&	   	& \circ  &   & \circ &    \circ  }$
      
      \caption{Possible Runs of Process $\pi_1 \parallel \pi_2 \parallel \pi_3$}
      \end{figure}

      So :
    
      \begin{itemize}
      \item $ \langle \pi_{1 \parallel 2 \parallel 3} \rangle . ( K_2 p  \land  K_3 p) $ is true
      \item $ \langle \pi_{1 \parallel 2 \parallel 3} \rangle . ( K_2 p  \lor  K_3 p) $ is true
      \item $ \langle \pi_{1 \parallel 2 \parallel 3} \rangle .  \neg ( K_2 p  \lor  K_3 p) $ is false
      \end{itemize}
    
    \end{example}

  \subsection{Bisimulation}
    
    The concept of bisimulation is a key notion in any process algebra. It is an equivalence relation between processes which have mutually similar behavior. The intuition is that two bisimilar processes cannot be distinguished by an external observer. Using the notion of bisimulation allows us to transform any process in an equivalent one that is  a summation of all their possible actions, that is what the Expansion Law (theorem \ref{exp_law}) states.
    
    There are two possible semantics for the $\tau$ action in CCS: it can be regarded as being observable, in the same way as the communication actions, or it can be regarded as being invisible. We adopt the first one, since it is more generic and fits better in our formalism. Whenever the $\tau$ action is observable the bisimulation relation is called {\it strong}.
    
    \begin{defn}[\cite{milner89}]\label{bisimula}
      Let $\Pi$ be the set of all  processes. A set $Z \subseteq \Pi \times \Pi$ is a \emph{strong bisimulation} if $(\pi_1,\pi_2) \in Z$ implies the following for all $\alpha \in $ AMS :
      \begin{itemize}
        \item If $\pi_1 \stackrel{\alpha}{\rightarrow} \pi_1'$, then there is $\pi_2' \in \Pi$ such that $\pi_2 \stackrel{\alpha}{\rightarrow} \pi_2'$ and $(\pi_1',\pi_2') \in Z$;
        \item If $\pi_2 \stackrel{\alpha}{\rightarrow} \pi_2'$, then there is $\pi_1' \in \Pi$ such that $\pi_1 \stackrel{\alpha}{\rightarrow} \pi_1'$ and $(\pi_1',\pi_2') \in Z$;
        \item $\pi_1 \stackrel{\alpha}{\rightarrow} \surd$ if and only if $\pi_2 \stackrel{\alpha}{\rightarrow} \surd$.
      \end{itemize}
    \end{defn}
        
    \begin{defn}[\cite{milner89}]\label{def:bisim}
      Two process $\pi$ and $\pi'$ are \emph{strongly bisimilar} (or simply \emph{bisimilar}), denoted by $\pi \simeq \pi'$, if there is a strong bisimulation $Z$ such that $(\pi,\pi') \in Z$.
    \end{defn}
    
    Now, we introduce the Expansion Law, which is very important in the definition of the semantic and in the axiomatization of our logic. We present a particular case of the Expansion Law, which is suited to our needs. The most general case of the Expansion Law is presented in \cite{milner89}.
    
    \begin{thm}[\cite{milner89}][Expansion Law (EL)]\label{teo:EL} \label{exp_law}
      Let $\pi =(\eta_1 \parallel ...  \parallel \eta_n)$. Then
      
      \[
      \pi \sim \sum_{\eta_i \stackrel{\alpha}{\rightarrow} \eta_i'} 
      \alpha.(\eta_1 \parallel ... \parallel \eta_i' \parallel  ... \parallel \eta_n) +  \sum_{(\eta_i \stackrel{c}{\rightarrow} \eta_i') \& (\eta_j \stackrel{\overline{c}}{\rightarrow} \eta_j')} \tau_{i,j}(.) . (\eta_1 \parallel ... \parallel \eta_i' \parallel ... \parallel \eta_j' \parallel ... \parallel \eta_n) \]
        
      were $\alpha$ is a action model and $\tau_{i,j}$ is a private action model resulted by the combination of two communication actions.
        
      We denote the right side of this bisimilarity by $Exp(\pi)$. We also denote by $\bf {0}$ the processes whose expansion is empty, i.e., there is no $(\eta_i \stackrel{c}{\rightarrow} \eta_i')$ , $(\eta_j \stackrel{\overline{c}}{\rightarrow} \eta_j')$ and $(\eta_k \stackrel{\alpha}{\rightarrow} \eta_k')$ for any $i,j,k \in \{1,...,n\}$.
    \end{thm} 
    
    \begin{proof}
      This follows from table \ref{tab:semccs} and definitions \ref{bisimula} and \ref{def:bisim}. A detailed proof for the most general case of this theorem can be found in \cite{milner89}.
    \end{proof}
    
    The Expansion Law is a very useful property of CCS processes. Its intuition is that processes can be rewritten as a summation 
    of all their possible actions. Suppose we have a processes $A \stackrel{def}{=} c.A' + \alpha.A''$ and $B \stackrel{def}{=} \overline{c}.B' + \beta.B''$, 
    then the process
    $(A \parallel B)$ is equivalent, using the Expansion Law, to $$(A \parallel B) \simeq \alpha.(A'' \parallel B) + \beta.(A \parallel B'') + \tau_{AB}.(A' \parallel B')$$

  \subsection{Language}
    
    In this section we present the DELWCA language.
    
    \begin{defn}\label{def-langcdel} 
      The DELWCA language consists of a set $\Phi$ of countably many proposition symbols, a set $\Pi$ of  programs as defined in \ref{program-syntax}, a finite set ${\cal A}$ of agents,  the boolean connectives $\neg$ and $\land$, a modality $\langle \pi \rangle$ for every program $\pi \in \Pi$ (as defined in section \ref{pc}) and a modality $K_a$ for each agent $a$. The formulas are defined as follows:
      \[
        \varphi ::= p \mid \top \mid \neg \varphi \mid \varphi_1 \wedge \varphi_2 \mid \langle \pi \rangle \varphi \mid K_i \varphi 
      \]
      where $p \in \Phi$, $\pi \in \Pi$, $i \in {\cal A}$ and $\langle \pi \rangle \varphi$ means that exists a execution of $\pi$ that leads to a state where $\varphi$ is true.
    \end{defn}

    \subsection{Semantics}\label{sem-cdel}    
      For communication actions (actions in ${\cal N}$) we need to relax the fact that relations in action models are equivalence relations, we just need them to be relations. For this case all the definitions of action models (def. \ref{action-models}), execution (product) of action models (def. \ref{product-action-models}), composition of action models (def. \ref{comp-action-models}) can be easily adapted.
      
    \begin{defn}\label{tau} 
      Let ${\cal A}$ be the set of all agents and $i,j \in {\cal A}$. The  action model $\tau^{s}_{i,j}(\varphi) = (\amodel,\actiona)$, with $\amodel = \langle \Actions, \sim, \pre \rangle$, is defined as follows:    
      \begin{itemize}
        \item $\Actions=\{\mathsf{\actiona, \actionb}\}$
        \item $\sim_i \ =\{(\actiona,\actiona), (\actionb,\actionb)\}$
        \item $\sim_j \ =\{(\actiona,\actiona), (\actionb,\actionb) \}$
        \item $\sim_k \ =\{(\actiona, \actionb), (\actionb,\actionb) \}$, for all $k \in {\cal A}\backslash  \{i,j\}$
        \item $\pre(\mathsf{\actiona})= \varphi$
        \item $\pre(\mathsf{t})= \top$
      \end{itemize}    
    \end{defn}

    \begin{figure}[H]
      \begin{center}
      \includegraphics[scale=0.25]{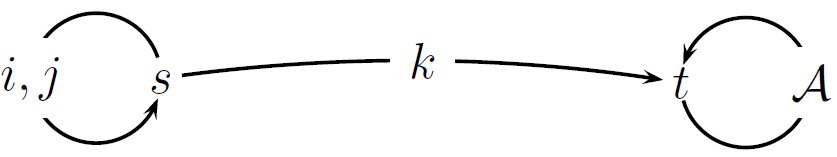}   
      \end{center}    
    \caption{Action Model for $\tau^{s}_{i,j}$}
    \end{figure}
    
    In order to obtain the definition of satisfaction for DELWCA  we must add the following condition to  definition \ref{satisfaction-am}:
    
    \vspace{0.3cm}
    
    \noindent $\llbracket (\eta_1 \parallel ...  \parallel \eta_n) \rrbracket =   $
    $ ~  \{ ~\llbracket\tau_{i,j}(.) \rrbracket ;
    \llbracket(\eta_1 \parallel ... \parallel \eta_i' \parallel ... \parallel \eta_j' \parallel ... \parallel \eta_n)\rrbracket$, 
    for all
    $(\eta_i \stackrel{\overline{c_1}}{\rightarrow} \eta_i') ~~\&$\\$ (\eta_j \stackrel{c_1}{\rightarrow} \eta_j')~\}~  \bigcup ~ \{~ \llbracket \alpha \rrbracket ; 
  \llbracket(\eta_1 \parallel ... \parallel \eta_i' \parallel  ... \parallel \eta_n)\rrbracket$,  
    for all
    $(\eta_i \stackrel{\alpha}{\rightarrow} \eta_i')~\} $
        
    \subsubsection{Axiomatization}\label{ax-cdel}
        
    \begin{enumerate}
    
      \item  {\it All instantiations of propositional tautologies},
    
      \bigskip
        
      \centerline{\bf Epistemic Logic Axioms\/}
    
      \bigskip
    
      \item[]  {\it Axioms (i), (ii), (iii), (iv) and (v) of  section \ref{mael-ax}},
    
      \bigskip
    
      \centerline{\bf Action Model  Axioms\/}
    
      \bigskip
    
      \item[]  {\it Axioms (vi), (vii), (viii) and (ix) of  section \ref{ax-del}},
  
      \bigskip
    
      \centerline{\bf PDL  Axioms\/}
    
      \bigskip
    
      \item[(x)] $[\pi](\phi \rightarrow \psi) \rightarrow ([\pi] \phi \rightarrow [\pi] \psi)$ (K axiom)
        
      \item[(xi)] $[\pi_1] [\pi_2] \phi \eq [\pi_1;\pi_2] \phi$ (Composition)
      
      \item[(xii)] $[\pi_1 + \pi_2]  \phi \eq [\pi_1] \phi \land  [\pi_2] \phi$ (Non-deterministic Choice)
      
      \item[(xiii)] $[\alpha.\pi] \phi \eq [\alpha][\pi]\phi$ (Prefix)\footnote{It is important to notice that Prefix is a special case of Composition}
      
      \item[(xiv)] $[\alpha.\pi] \phi \eq pre(\alpha) \rightarrow [\pi]\phi$ 
    
      \bigskip
    
      \centerline{\bf Concurrent Action  Axiom\/}
    
      \bigskip
    
      \item[(xv)] $[\eta_1 \parallel ...  \parallel \eta_n] \phi \eq [Exp(\eta_1 \parallel ...  \parallel \eta_n)] \phi$ \label{concurrentActionAxiom}     
       
    \end{enumerate}
    
    \centerline{\bf Inference Rules\/}
    
    \smallskip    
    M.P. $\varphi, \varphi\rightarrow \psi / \psi$
    \qquad U.G. $\varphi/ [ \pi ]\varphi$ \qquad $\varphi/ K_a \varphi$  
 
  \bigskip  
  \begin{proposition}
      $\vdash [\alpha;\pi_2] \phi \eq [\alpha][\pi_2] \phi\eq [\alpha.\pi_2] \phi \eq pre(\alpha) \rightarrow [\pi_2]\phi$
  \end{proposition}
  \bigskip

  \begin{example}
    
    A supervisor Ane (1) and her two students Bob(2) and Cathy(3) are working in their computer located at their own house. The supervisor wants to book a meeting "tomorrow at 16:00". She sends a message asynchronously to Bob and Cathy. We are supposing that the supervisor uses channels $c_{12}$ and $c_{13}$ to communicate with Bob and Cathy respectively. We represent Anne, Bob and Cathy by processes $\pi_1$, $\pi_2$ and $\pi_3$ respectively, and their parallel composition by $\pi_{1\parallel 2 \parallel 3}$.  
    
    \begin{itemize}    
      \item $\pi_1 = \overline{c}_{12}(p);\overline{c}_{13}(p) + \overline{c}_{13}(p);\overline{c}_{12}(p)$      
      \item $\pi_2 = c_{12}(.)$      
      \item $\pi_3 = c_{13}(.)$      
      \item $\pi_{1\parallel 2 \parallel 3} = (\pi_1 \parallel \pi_2 \parallel \pi_3)$    
    \end{itemize}    
 
    We have two possible runs process $\pi_{1\parallel 2 \parallel 3}$ as shown in the tree in figure \ref{exemplo2runs}.
    
    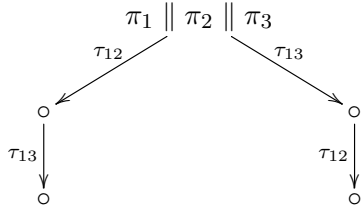
\begin{figure}[H]
    \centering    
      $\xymatrix{
      &		 	&\pi_1 \parallel \pi_2 \parallel \pi_3 \ar[ld]_{\tau_{12}} \ar[rd]^{\tau_{13}}	&						\\
      &\circ \ar[d]_{\tau_{13}}	& 				&\circ	\ar[d]_{\tau_{12}}&			\\
      &\circ	 		&				&	\circ				}$    
      \caption{Possible Runs of Process $\pi_1 \parallel \pi_2 \parallel \pi_3$}
      \label{exemplo2runs}
    \end{figure}
    
    Let propositional symbol $p$ represent "tomorrow at 16:00". The epistemic model $\cal{M}_0$ at the begging is as shown in figure \ref{initialEpistemicModel}.
     
    \begin{figure}[H]
      \centering
    
      $\xymatrix{
      &u ~ \circ \ar@{-}[r]	& 	2,3		\ar@{-}[r]	&\circ ~ v &			\\
      &p 		&				&	\neg p				}$
      \caption{Initial Epistemic Model ${\cal M}_0$}
      \label{initialEpistemicModel}
    \end{figure}
    
    The action models for $\tau_{12}$ and $\tau_{13}$ are presented in figures \ref{actionModelTau12} and \ref{actionModelTau13}.
    \begin{figure}[H]
      \begin{center}
        \includegraphics[scale=0.25]{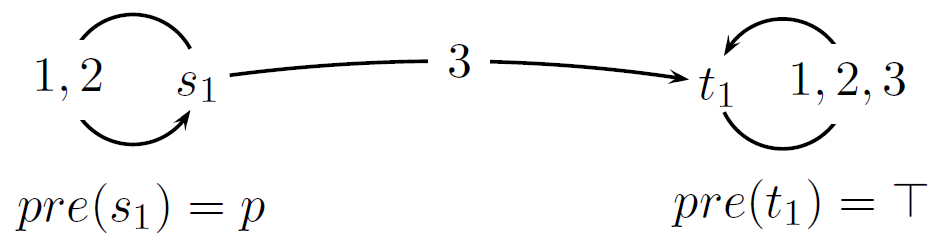}  
      \end{center}
      \caption{Action Model for $\tau_{12}$}
      \label{actionModelTau12}
    \end{figure}
    
    \begin{figure}[H]
      \begin{center}
        \includegraphics[scale=0.25]{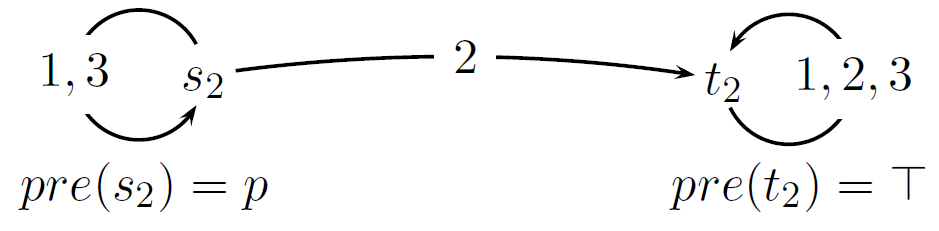}
      \end{center}
    \caption{Action Model for $\tau_{13}$}
    \label{actionModelTau13}
    \end{figure}

    Suppose $\tau_{12}$ is performed before $\tau_{13}$. After the execution of $\tau_{12}$ we obtain the epistemic model picture in figure \ref{epistmeicModelM0Tau12}.  
    
    \begin{figure}[H]
      \begin{center}
        \includegraphics[scale=0.25]{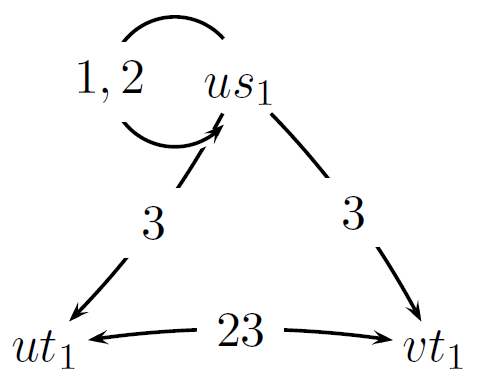}  
      \end{center}
    \caption{Epistemic Model ${\cal M}_{1} = {\cal M}_{0} \otimes \tau_{12}$}
    \label{epistmeicModelM0Tau12}
    \end{figure}

    It is important to notice that at state $us_1$ Ane and Bob knows $p$ ${\cal M}_1 , us_1 \vdash K_1 p \land K_2 p$ but Cath doesn't  ${\cal M}_1 , us_1 \vdash \neg K_3 p$. After the second communication $\tau_{13}$ we have the epistemic model of figure \ref{epistemicModelM1Tau13}.

    \begin{figure}[H]
      \begin{center}
        \includegraphics[scale=0.25]{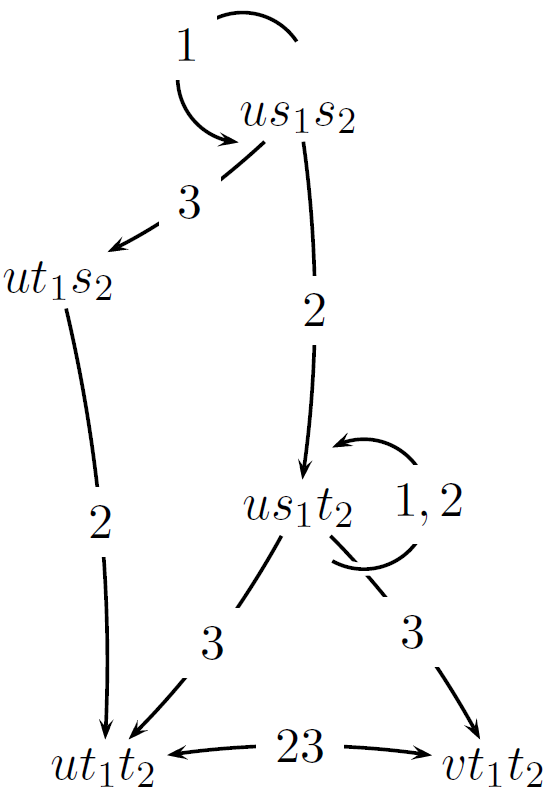}
      \end{center}
      \caption{Epistemic Model ${\cal M}_{2} = {\cal M}_{1} \otimes \tau_{13} =  {\cal M}_{0} \otimes \tau_{12} \otimes \tau_{13}$}
      \label{epistemicModelM1Tau13}
    \end{figure}
    
    We can notice, from figure \ref{epistemicModelM1Tau13} that at state $us_1s_2$ Ane, Bob and Cath knows $p$ ${\cal M}_2 , us_1s_2 \vdash K_1 p \land K_2 p \land K_3 p$ as expected. If we execute run $\tau_{13};\tau{12}$ we obtain the model ${\cal M}_{3}$ as shown in figure \ref{epismitecModelM3}.
    
    \vspace{1cm}
    
    \begin{figure}[H]
      \begin{center}
        \includegraphics[scale=0.25]{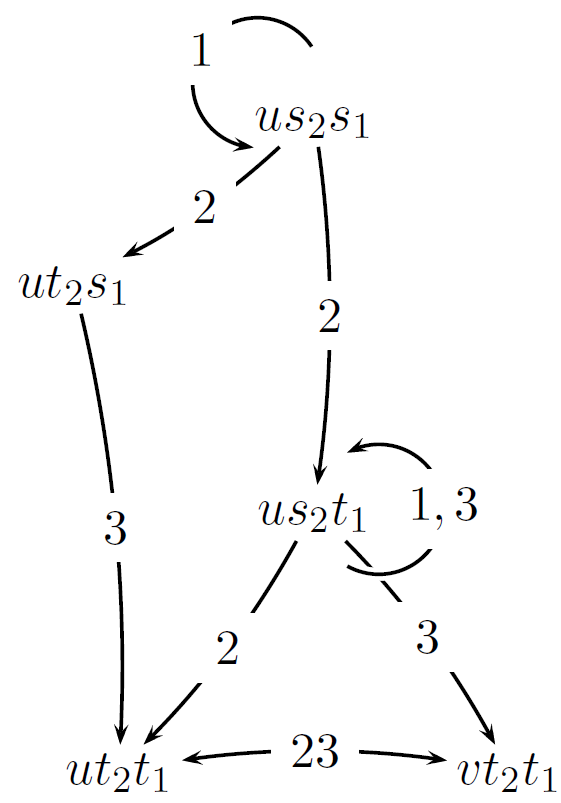}
      \end{center}
    \caption{Epistemic Model ${\cal M}_{3} =  {\cal M}_{0} \otimes \tau_{13} \otimes \tau_{12}$}
    \label{epismitecModelM3}
    \end{figure}
    
    We can show, from figure \ref{epismitecModelM3}, that Ane, Bob and Cath know $p$ ${\cal M}_3 , us_1s_2 \vdash K_1 p \land K_2 p \land K_3 p$ as expected.
  \end{example}
  \subsection{Soundness, Completeness and Decidability}\label{sat-pddel}

    \subsubsection{Soundness}
 
      We need to prove that all axioms are valid. Axioms $i$ to $xiii$ are standard from Dynamic Epistemic Logic literature and can be found in \cite{hwb}. We prove validity only for axiom \ref{concurrentActionAxiom}.
  
        \begin{lemma} $[\eta_1 \parallel ...  \parallel \eta_n] \phi \eq [Exp(\eta_1 \parallel ...  \parallel \eta_n)] \phi$ is valid.
        \label{soundnessLemma}
        \end{lemma}

    \subsubsection{Completeness}

      The proof of completeness is similar to the proof for Public Announcement and Action Models Logics introduced in \cite{DBLP:journals/iandc/BenthemEK06} Dynamic Epistemic Logic. We prove completeness showing that every formula in DELWCA is equivalent to formula in Epistemic Logic. In order to achieve that we only have to provide a translation function that translate every DELWCA formula to a formula without communication actions and concurrency.

    \subsubsection{Decidability}

      Decidability follows directly from the decidability of $ {\bf S5_a}$.
  
\section{Conclusions}
  
  In this work we present a Dynamic Epistemic Logic with Communication Actions that can be performed concurrently. In order to achieve that we propose a PDL like language for actions and develop a small process calculus.  We show that it's easy to model problems of communication and concurrency with the proposed dynamic epistemic logic. The main feature of it is the Expansion rule which allows for representing the parallel composition operator. This approach is similar to the one introduced in \cite{BS08,BS10}. 
  
  We represent communication actions as private Action Models where the relations are not equivalence relations. We present an axiomatization and prove completeness using reduction technique.
 
  As future work we would like to investigate the extension with common knowledge and/or iteration operators, study other types of communications where agents are not reliable or not trustful, extend this to Dynamic Epistemic Logic With Post-Conditions and change DEMO, or create a new Model Checker, to deal with concurrency and communication.

\bibliographystyle{entcs}
\bibliography{references}


\begin{appendices}
  
  \section{Soundness Proof}
    We need to prove Lemma \ref{soundnessLemma}.

    \begin{proof} 
      We have to show that $(1) \eq (2)$, where $(1)$ is $[\eta_1 \parallel ...  \parallel \eta_n] \phi$ and $(2)$ is $[Exp(\eta_1 \parallel ...  \parallel \eta_n)] \phi$.
      
      We can represent $(1)$ and $(2)$ like this:
      \begin{itemize}
      \item[] $(1) = (s | \forall y(s,s') \in \llbracket \eta_1 \parallel ...  \parallel \eta_n \rrbracket \imp s' \in \llbracket \phi \rrbracket )$
      \item[] $(2) = (s | \forall y(s,s') \in \llbracket Exp (\eta_1 \parallel ...  \parallel \eta_n) \rrbracket \imp s' \in \llbracket \phi \rrbracket )$
      \end{itemize}
      So we need to show that 
      \begin{itemize}
      \item[] $\llbracket \eta_1 \parallel ...  \parallel \eta_n \rrbracket \eq \llbracket Exp (\eta_1 \parallel ...  \parallel \eta_n) \rrbracket$ . $(3) \eq (4)$.
      \end{itemize}

      Using the definition \ref{satisfaction-am} we have that
      \begin{itemize}
      \item[] $(3)$ = $ \llbracket \tau \rrbracket ; \llbracket \eta_1 \parallel ... \eta^{'}_i ... \eta^{'}_j ... \parallel \eta_n \rrbracket \bigcup \llbracket \alpha \rrbracket ; \llbracket \eta_1 \parallel ... \eta^{'}_i ... \parallel \eta_n \rrbracket \bigcup  \llbracket \beta \rrbracket ; \llbracket \eta_1 \parallel ... \eta^{'}_j ... \parallel \eta_n \rrbracket  $
      \end{itemize}
      
      Using the expansion law we have that 
      \begin{itemize}
      
        \item[] $(4)$ = $\llbracket \sum_{(\eta_i \stackrel{c}{\rightarrow} \eta_i') \& (\eta_j \stackrel{\overline{c}}{\rightarrow} \eta_j')} \tau_{i,j}(.) . 
        (\eta_1 \parallel ... \parallel \eta_i' \parallel ... \parallel \eta_j' \parallel ... \parallel \eta_n) \\ + \sum_{(\eta_k \stackrel{\alpha}{\rightarrow} \eta_k')} \alpha_k . (\eta_1 \parallel ... \parallel \eta_k' \parallel ... \parallel \eta_n) + \sum_{(\eta_k \stackrel{\beta}{\rightarrow} \eta_k')} \beta_k . (\eta_1 \parallel ... \parallel \eta_k' \parallel ... \parallel \eta_n) \rrbracket$
        \item[] $(4)$ = $\llbracket \sum_{(\eta_i \stackrel{c}{\rightarrow} \eta_i') \& (\eta_j \stackrel{\overline{c}}{\rightarrow} \eta_j')} \tau_{i,j}(.) . 
        (\eta_1 \parallel ... \parallel \eta_i' \parallel ... \parallel \eta_j' \parallel ... \parallel \eta_n) \rrbracket \\ \bigcup \llbracket \sum_{(\eta_k \stackrel{\alpha}{\rightarrow} \eta_k')} \alpha_k . (\eta_1 \parallel ... \parallel \eta_k' \parallel ... \parallel \eta_n) \rrbracket \bigcup \llbracket \sum_{(\eta_k \stackrel{\beta}{\rightarrow} \eta_k')} \beta_k . (\eta_1 \parallel ... \parallel \eta_k' \parallel ... \parallel \eta_n) \rrbracket$
        
        \item[] Since in $(2)$ we are using $(s,s')$ we can omite the $\sum$

        \item[] $(4)$ = $\llbracket \tau_{i,j}(.) . 
        (\eta_1 \parallel ... \parallel \eta_i' \parallel ... \parallel \eta_j' \parallel ... \parallel \eta_n) \rrbracket \\ \bigcup \llbracket \alpha_k . (\eta_1 \parallel ... \parallel \eta_k' \parallel ... \parallel \eta_n) \rrbracket \bigcup \llbracket \beta_k . (\eta_1 \parallel ... \parallel \eta_k' \parallel ... \parallel \eta_n) \rrbracket$
        
        \item[] Using definition 9 we have
        \item[] $(4)$ = $\llbracket \tau_{i,j}(.) \rrbracket ; \llbracket 
        (\eta_1 \parallel ... \parallel \eta_i' \parallel ... \parallel \eta_j' \parallel ... \parallel \eta_n) \rrbracket \\ \bigcup \llbracket \alpha_k \rrbracket ; \llbracket (\eta_1 \parallel ... \parallel \eta_k' \parallel ... \parallel \eta_n) \rrbracket \bigcup \llbracket \beta_k \rrbracket ; \llbracket (\eta_1 \parallel ... \parallel \eta_k' \parallel ... \parallel \eta_n) \rrbracket$
        
      \end{itemize}

      So $(1) \eq (2)$
      
    \end{proof}
  \clearpage
  \section{Completeness Proof}
     We need to provide a translation function that translate every DELWCA formula to a formula without communication actions. 

    \begin{definition}\label{traducaoAM}     
      The translation function t: ${\cal L}_{DELWCA} \to {\cal L}_K$ is defined as follows: 
          
      \[ \begin{array}{lcl} 
        t(p) & = & p \\ 
        t(\neg \varphi) & =  & \neg t(\varphi) \\ 
        t(\varphi \land \psi) & = & t(\varphi) \land t(\psi) \\ 
        t(K_a \varphi) & = & K_a t(\varphi) \\ 
        t( [\amodel,\actiona] p ) & = & t ( \pre(\actiona)  \to p )) \\ 
        t( [\amodel,\actiona] \neg \varphi ) & = & t ( \pre(\actiona) \to \neg [\amodel,\actiona] \varphi  ) \\ 
        t( [\amodel,\actiona] (\varphi \land \psi) ) & = & t ( [\amodel,\actiona] \varphi  \land  [\amodel,\actiona] \psi  ) \\ 
        t( [\amodel,\actiona] K_a \varphi ) & = & t ( \pre(\actiona) \to K_a [\amodel,\actiona] \varphi  ) \\ 
        t( [\pi_1] [\pi_2] \varphi ) & = & t ( [\pi_1;\pi_2] \varphi ) \\ 
        t([\pi_1 + \pi_2] \varphi) & = & t([\pi_1]\varphi) \land t([\pi_2]\varphi) \\ 
        t( [\alpha] [\pi] \varphi ) & = & t ( [\alpha;\pi_1] \varphi ) \\ 
          
        t([\eta_1 \parallel ...  \parallel \eta_n] \varphi) & = & t([Exp(\eta_1 \parallel ...  \parallel \eta_n)] \varphi) \\ 
          
        t([{\bf 0}] \varphi) & = & t( \varphi) \\ 
        t( [\amodel,\actiona] [\pi] \varphi ) & = & t ( \pre(\actiona)  \to [\pi] \varphi ) \\ 
      \end{array} \] 
          
      \end{definition} 
        
      In order to prove completeness we need to prove that every CDEL formula can be proved (in the axiomatic system) equivalent to its translation. This proof is by induction on the {\it complexity} of each formula which is defined below. 
        
    Let $k$ be the number of all possible communications that can occur in $\eta_1 \parallel ...  \parallel \eta_n$, i.e. all pairs $(\eta_i \stackrel{s_{i,j}(.)}{\rightarrow} \eta_i')$ and $ (\eta_j \stackrel{r_{i,j}(.)}{\rightarrow} \eta_j')$, for $1 \leq i,j \leq n$. 
    
    \begin{definition}\label{complexidadeAM} 
    
      The complexity $c : {\cal L}_{DELWCA} \to \nat $ is defined as follows: 
          
      \[ \begin{array}{lcl} 
          
        c({\bf 0}) & = & 1 \\ 
        c(\alpha;\pi) & = & 1 + c(\pi)\\ 
        c(\pi_1 + \pi_2) & = & 2 + max\{c(\pi_1),c(\pi_2)\}\\ 
        c(\pi_1 ; \pi_2) & = & 1 + c(\pi_1) + c(\pi_2)\\ 
        c(\eta_1 \parallel ...  \parallel \eta_n) & = & k + c(\eta_1) + \cdots + c(\eta_n)\\ 
        &&\\ 
        c(p) & = & 1 \\ 
        c(\neg \varphi) & =  & 1 + c(\varphi) \\ 
        c(\varphi \land \psi) & = & 1 + max ( c(\varphi), c(\psi) ) \\ 
        c(K_a \varphi) & = & 1 + c(\varphi) \\ 
        c([\pi] \varphi) & = & (4 + c(\pi)) * c(\varphi) \\ 
        c(\amodel,\actiona) & = & max \{ c( \pre(t)) \mid t \in \amodel \}\\ 
        c([{\bf 0}]) \varphi) & = & 1+ c( \varphi) \\ 
          
      \end{array} \]

    \end{definition}

    We have to prove that the complexity of a formula is strictly greater than the complexity of its translation. In order to achieve that we prove the next lemma that assures that after a communication the complexity always decreases. 
    
    \begin{lemma} \label{comm} If $\pi \stackrel{\alpha}{\rightarrow} \pi'$, then $c(\pi) > c(\pi')$. 
      
    \end{lemma} 
     
    \begin{proof} By induction on $| \pi |$. 
      \noindent Base: It holds for ${\bf 0}$, once ${\bf 0}  \stackrel{\alpha}{\not \rightarrow} $, 
      
      \noindent Induction Hypothesis: it holds for $| \pi | < m$. 
        
      \begin{enumerate} 
        
        \item $\pi = \alpha. \pi'$: we know that $\pi \stackrel{\alpha}{\rightarrow} \pi'$ and $c(\alpha. \pi) = 1+ c(\pi')$, so $c(\pi) > c(\pi')$. 
          
        \item $\pi = \pi_1 + \pi_2$: either $\pi_1 \stackrel{\alpha}{\rightarrow} \pi_1'$ or $\pi_2 \stackrel{\alpha}{\rightarrow} \pi_2'$. By the semantics rules of table \ref{tab:semccs} either $\pi \stackrel{\alpha}{\rightarrow} \pi_1' $ (1) or $\pi \stackrel{\alpha}{\rightarrow}  \pi_2'$ (2). By the induction hypothesis $c(\pi_1) > c(\pi_1')$ and $c(\pi_2) > c(\pi_2')$. As 
          
        $c(\pi) = 2 + max\{\pi_1,\pi_2\}$ 
          
        From (1) $c(\pi) >  c(\pi_1) > c(\pi_1')$ 
          
        From (2) $c(\pi) >  c(\pi_2) > c(\pi_2')$ 
          
        \item $\pi = \pi_1 ; \pi_2$: if $\pi_1 \stackrel{\alpha}{\rightarrow} \pi_1'$, then, by the semantics rules of table \ref{tab:semccs}, $\pi \stackrel{\alpha}{\rightarrow} \pi_1' ; \pi_2$. By the induction hypothesis $c(\pi_1) > c(\pi_1')$. As 
          
        $c(\pi) = 1 + c(\pi_1) + c(\pi_2) > 1 + c(\pi_1') + c(\pi_2) = c(\pi_1' ; \pi_2)$ 
          
        \item $\pi = \eta_1 \parallel ...  \parallel \eta_n$: if $\pi \stackrel{\tau_{i,j}(.)}{\rightarrow} \pi'$, then there exists 
        $(\eta_i \stackrel{s_{i,j}(.)}{\rightarrow} \eta_i')$ and $ (\eta_j \stackrel{r_{i,j}(.)}{\rightarrow} \eta_j')$, for $1 \leq i,j \leq n$. And $\pi' = \eta_1 \parallel ... \parallel \eta_i' \parallel ... \parallel \eta_j' \parallel ... \parallel \eta_n$. By the induction hypothesis $c(\eta_i) > c(\eta_i')$ and $c(\eta_j) > c(\eta_j')$. So, 
          
        $c(\pi)  =  k + c(\eta_1) + \cdots + c(\eta_i) + \cdots + c(\eta_j) +\cdots + c(\eta_n) >  k + c(\eta_1) + \cdots + c(\eta_i') + \cdots + c(\eta_j') +\cdots + c(\eta_n) = c(\pi')$.        
      \end{enumerate}
    \end{proof}

    \begin{cor}\label{comp-exp} $c(\pi = \eta_1 \parallel ...  \parallel \eta_n) > c(Exp(\eta_1 \parallel ...  \parallel \eta_n)$     
    \end{cor}

    \begin{proof} By the definition of $Exp$  
    
      $c(Exp(\eta_1 \parallel ...  \parallel \eta_n)  = $ \\ 
      $c( \sum \tau_{i,j} . (\eta_1 \parallel ... \parallel \eta_i' \parallel ... \parallel \eta_j' \parallel ... \parallel \eta_n)) =$\\ 
      $k-1 + max_{ij}(c(\tau_{i,j} . (\eta_1 \parallel ... \parallel \eta_i' \parallel ... \parallel \eta_j' \parallel ... \parallel \eta_n)) =$\\ 
      $k + max_{ij}(c( (\eta_1 \parallel ... \parallel \eta_i' \parallel ... \parallel \eta_j' \parallel ... \parallel \eta_n)) =$ 
      We know that  
        
      $\eta_1 \parallel ...  \parallel \eta_n \stackrel{\tau_{i,j}(.)}{\rightarrow} \eta_1 \parallel ... \parallel \eta_i' \parallel ... \parallel \eta_j' \parallel ... \parallel \eta_n$, 
        
      by lemma \ref{comm} 
        
      $c(\eta_1 \parallel ...  \parallel \eta_n) > c( \eta_1 \parallel ... \parallel \eta_i' \parallel ... \parallel \eta_j' \parallel ... \parallel \eta_n)$, Then $c(\pi = \eta_1 \parallel ...  \parallel \eta_n) > c(Exp(\eta_1 \parallel ...  \parallel \eta_n)$. 
    \end{proof}
    
    \begin{lemma}\label{complexidadeAMLema} 
      For all $\varphi$ and $ \psi$: 
              
      \begin{enumerate} 
      \item $c(\psi) \geq c(\varphi)$ se  $\varphi \in Sub(\psi)$  
      \item $c([\amodel,\actiona]p) > c(\pre(\action)  \to p))$ 
      \item $c([\amodel,\actiona]\neg \varphi) > c(\pre(\action) \to \neg [\amodel,\actiona] \varphi )$ 
      \item $c([\amodel,\actiona] (\varphi \land \psi) ) > c( [\amodel,\actiona] \varphi \land [\amodel,\actiona] \psi )$ 
      \item $c([\amodel,\actiona] K_a \varphi) > c(\pre(\action) \to K_a [\amodel,\actiona] \varphi )$ 
      \item $c([\pi_1] [\pi_2] \varphi) > c([\pi_1 ; \pi_2] \varphi )$ 
      \item $c([\pi_1 + \pi_2] \varphi ) > c( [\pi_1] \varphi \land [\pi_2] \varphi)$ 
      \item $c( [\alpha] [\pi] \varphi ) >  c( [\alpha;\pi_1] \varphi )$ 
        
      \item $c([\eta_1 \parallel ...  \parallel \eta_n] \varphi) > c([Exp(\eta_1 \parallel ...  \parallel \eta_n)] \varphi)$ 
              
      \item $c([{\bf 0}] \varphi) > c( \varphi) $ 
      \end{enumerate}            
              
    \end{lemma}

    \begin{proof} The proofs of  1, 3, 4 e 5 is straightforward from definition \ref{complexidadeAM}.    
        \begin{itemize}
          \item[2.] We know that  $c(\phi \to \psi) = 2 + c(\phi) + c(\psi)$ 
            \begin{itemize} 
              \item $c([\amodel,\actiona]p) =  (4 + c(\amodel,\actiona)) * c(p) = 4 + max \{ c( \pre(t) \mid t \in \amodel \}$\\ 
              
              \item $c(\pre(\action)  \to p) = 2 + c(\pre(\action) ) +  1 = 3 + c(\pre(\action) ) $\\ 
                
            \end{itemize} 
                
            Therefore, $c([\amodel,\actiona]p) > c(\pre(\action)  \to p))$
          \item[6.]  $c([\pi_1 ; \pi_2] \varphi ) = (4+c(\pi_1 ; \pi_2).c(\varphi)) $ 
    
              $c([\pi_1 ; \pi_2] \varphi ) = (5+c(\pi_1) + c(\pi_2)).c(\varphi))$ 
              
              $c([\pi_1] [\pi_2] \varphi) =   (4+c(\pi_1)) . c([\pi_2]\varphi)$ 
              
              $c([\pi_1] [\pi_2] \varphi) =   (4+c(\pi_1)) . (4+c(\pi_2)).c(\varphi))$ 
              
              $c([\pi_1] [\pi_2] \varphi) =   (16+ 4(c(\pi_1) + c(\pi_2)) + c(\pi_1) c(\pi_2))c(\varphi))$ 
              
              $c([\pi_1] [\pi_2] \varphi) >   c([\pi_1 ; \pi_2] \varphi )$   
          \item[7.]  $c([\pi_1  + \pi_2] \varphi ) = (4+c(\pi_1 + \pi_2).c(\varphi)) $ 
    
            $c([\pi_1 + \pi_2] \varphi ) = (6+max\{c(\pi_1) , c(\pi_2)\}).c(\varphi))$ 
               
            $c([\pi_1] \varphi \land  [\pi_2] \varphi) =   (1+max\{c([\pi_1] \varphi) , c([\pi_2] \varphi)\})$ 
                
            $c([\pi_1] \varphi \land  [\pi_2] \varphi) =   (1+max\{ (4+c(\pi_1))c( \varphi) , (4+c(\pi_2)) c(\varphi)\})$ 
                
            $c([\pi_1] \varphi \land  [\pi_2] \varphi) =   (5 + max\{ c(\pi_1) , c(\pi_2)\}) c(\varphi)$ 
                
            $c([\pi_1  + \pi_2] \varphi ) >  c([\pi_1] \varphi \land  [\pi_2] \varphi) $ 
                
          \item[8.] this case is analogous to 6. 
          \item[9.]  
    
            \begin{itemize} 
              \item $c([\eta_1 \parallel ...  \parallel \eta_n] \varphi)) =  (4 + c(\eta_1 \parallel ...  \parallel \eta_n) )* c(\varphi)$\\  
    
              \item $c([Exp(\eta_1 \parallel ...  \parallel \eta_n)] \varphi) = (4 + c(Exp(\eta_1 \parallel ...  \parallel \eta_n)) )* c(\varphi)$\\  
              
              The proof of completeness is similar to the proof for Public Announcement and Action Models Logics introduced in \cite{DBLP:journals/iandc/BenthemEK06} Dynamic Epistemic Logic. We prove completeness showing that every formula in DELWCA is equivalent to formula in Epistemic Logic. In order to achieve that we only have to provide a translation function that translate every DELWCA formula to a formula without communication actions.          
            \end{itemize}
            By corollary \ref{comp-exp},   
    
            $c([\eta_1 \parallel ...  \parallel \eta_n] \varphi) > c([Exp(\eta_1 \parallel ...  \parallel \eta_n)] \varphi)$ 
        \end{itemize}
    \end{proof}
    The following lemma asserts that every formula is deductively equivalent to its translation. 
    
    \begin{lemma}\label{equivalentetraducaoLema} 
    
      For all formulas $\varphi \in {\cal L}_{DELWCA} holds:$ 
    
        \[  \vdash \varphi \eq t(\varphi) \] 
    
    \end{lemma}

    \begin{proof} By induction on the complexity of $\varphi$ ($c(\varphi)$).

      \noindent {\bf Induction Hypothesis}: Suppose \[  \vdash \varphi \eq t(\varphi) \] holds for formulas $\varphi$ where $c(\varphi) < m$ 
        
      \begin{enumerate} 
        
        \item Base: $\varphi = p$ follows from the tautology $\vdash p \eq p $; 
        item $\varphi = \neg \psi, ~ \psi_1 \land \psi_2 ~ K_a \psi$: straightforward from the Induction Hypothesis;     
  
        \item $\varphi = [\amodel,\actiona]  p $:  
          
          $t( [\amodel,\actiona] p ) = t ( \pre(\actiona)  \to p ))$, by lemma \ref{complexidadeAMLema} and the induction hypothesis we have 
            
          $\vdash t ( \pre(\actiona)  \to p )) \eq  \pre(\actiona)  \to p )$, but by axiom 6 
            
          $\vdash  \pre(\actiona)  \to p ) \eq  [\amodel,\actiona]  p $, and thus 
            
          $\vdash t( [\amodel,\actiona] p ) \eq  [\amodel,\actiona] p$ 
          
        \item $\varphi = [\eta_1 \parallel ...  \parallel \eta_n] \psi$: 
          
          $t([\eta_1 \parallel ...  \parallel \eta_n] \psi) = t([Exp(\eta_1 \parallel ...  \parallel \eta_n)] \psi)$, by lemma \ref{complexidadeAMLema} and the induction hypothesis we have 
            
          $\vdash t([Exp(\eta_1 \parallel ...  \parallel \eta_n)] \psi) \eq [Exp(\eta_1 \parallel ...  \parallel \eta_n)] \psi$, by axiom 13 
            
          $\vdash [\eta_1 \parallel ...  \parallel \eta_n] \psi \eq [Exp(\eta_1 \parallel ...  \parallel \eta_n)] \psi$ and thus 
            
          $\vdash t([\eta_1 \parallel ...  \parallel \eta_n] \psi) \eq ([\eta_1 \parallel ...  \parallel \eta_n] \psi$ 
          
        \item $\varphi \in \{[\alpha.\pi]\psi, [\pi_1 + \pi_2]\psi, [\pi_1 ; \pi_2]\psi, [{\bf 0}] \psi \}$ is analogous to case 4. 
              
      \end{enumerate}
    \end{proof}
  
    Completeness follows.

    \begin{theorem}[Complteness]\label{completudeP}  
    
      For all $\varphi \in {\cal L}_{DELWCA}$ 
        
      \[  \models \varphi \text{ implies }  \vdash \varphi \] 
      
    \end{theorem}

    \begin{proof} Suppose $\models \varphi$. By lemma \ref{equivalentetraducaoLema} we know that 
      $\vdash \varphi \eq t(\varphi)$. By soundness we have $\models \varphi \eq t(\varphi)$ and thus $\models t(\varphi)$. But as 
      $t(\varphi)$ has no action modalities, it is a formula of  Multi-agent Epistemic Logic  ${\bf S5_a}$ and as  ${\bf S5_a}$ is complete we have 
      $\vdash_{S5_a} t(\varphi)$, but as  $ {\bf S5_a}$ is contained in   {\bf DELWCA}, we have $\vdash t(\varphi)$.

    \end{proof}

\end{appendices}

\end{document}